\documentclass[a4paper,UKenglish]{lipics-v2021}

\usepackage{microtype,todonotes}
\usepackage{amsthm,xcolor}
\usepackage{xypic}
\usepackage{tikz}
\usetikzlibrary{arrows}
\usetikzlibrary{shapes}

\usetikzlibrary{shapes}
\newcommand{\red}[1]{\textcolor{black}{#1}}

\def\romannum{\begingroup
  \def\theenumi{\textup{(\roman{enumi})}}%
  \def\p@enumi{}%
  \def\labelenumi{\theenumi}%
  \enumerate}

\newcommand{\CSP}{\ensuremath{\mathrm{CSP}}}

\newcommand{\NP}{\ensuremath{\mathrm{NP}}}

\newcommand{\coNP}{\ensuremath{\mathrm{co\mbox{-}NP}}}
\newcommand{\Pspace}{\ensuremath{\mathrm{Pspace}}}
\newcommand{\MC}{\ensuremath{\mathrm{MC}}}
\newcommand{\PMC}{\ensuremath{\mathrm{PMC}}}

\newcommand{\notmodels}{\ensuremath{ \models \hspace{-3mm} / \hspace{2mm} }}

\nolinenumbers

\title{Model-checking positive equality free logic on a fixed structure (direttissima)}

\author{Manuel Bodirsky}{TU Dresden, Germany}{Manuel.Bodirsky@tu-dresden.de}{}{}
\author{Marcin Kozik}{Jagiellonian University, Krak\'ow, Poland}{marcin.kozik@uj.edu.pl}{}{}
\author{Florent Madelaine}{LACL, Universit\'e Paris-Creteil, France}{florent.madelaine@u-pec.fr}{}{}
\author{Barnaby Martin}{Durham University, U.K.}{barnabymartin@durham.ac.uk}{}{}
\author{Micha{\l} Wrona}{Jagiellonian University, Krak\'ow, Poland}{michal.wrona@uj.edu.pl}{}{}
\authorrunning{M. Bodirsky et al.}
\Copyright{M. Bodirsky et al.}



\EventEditors{John Q. Open and Joan R. Acces}
\EventNoEds{2}
\EventLongTitle{42nd Conference on Very Important Topics (CVIT 2016)}
\EventShortTitle{CVIT 2016}
\EventAcronym{CVIT}
\EventYear{2016}
\EventDate{December 24--27, 2016}
\EventLocation{Little Whinging, United Kingdom}
\EventLogo{}
\SeriesVolume{42}
\ArticleNo{23}
\begin{document}
\maketitle

\begin{abstract}
  We give a new, direct proof of the tetrachotomy classification for the model-checking problem of positive equality-free logic parameterised by the model. The four complexity classes are Logspace, NP-complete, co-NP-complete and Pspace-complete. The previous proof of this result relied on notions from universal algebra and core-like structures called $U$-$X$-cores. This new proof uses only relations, and works for infinite structures also in the distinction between Logspace and NP-hard under Turing reductions. 

For finite domains, the membership in NP and co-NP follows from a simple argument, which breaks down already over an infinite set with a binary relation. We develop some interesting new algorithms to solve NP and co-NP membership for a variety of infinite structures. We begin with those first-order definable in $(\mathbb{Q};=)$, the so-called equality languages, then move to those first-order definable in $(\mathbb{Q};<)$, the so-called temporal languages. However, it is first-order expansions of the Random Graph $(V,E)$ that provide the most interesting examples. In all of these cases, the derived classification is a tetrachotomy between Logspace, NP-complete, co-NP-complete and Pspace-complete.
\end{abstract}


\section{Introduction}
\label{sec:intro}

The question of model-checking syntactic fragments of first-order logic on a fixed model $\mathcal{B}$ was discussed in \cite{MadelaineM18}. The syntactic fragments considered correspond to limiting which of the symbols $\{\forall,\exists,\wedge,\vee,\neg,=\}$ we permit. The most famous of the fragments for this task is probably \emph{primitive positive} logic, which has $\{\exists,\wedge\}$ and corresponds to the \emph{constraint satisfaction problem} (CSP). The fixing of the model corresponds to what Vardi called expression complexity in \cite{Vardi82} and what is known as non-uniform in the CSP literature \cite{FederVardi}, where the model is usually known as the template.

For the majority of the syntactic fragments, a complete classification of computational complexity is possible, as one varies the template $\mathcal{B}$, and this classification is simple to derive. Then some fragments are equivalent to others through de Morgan's laws (e.g. $\{\exists,\wedge\}$ is equivalent to $\{\forall,\vee\}$, modulo \NP-completeness morphing to \coNP-completeness). Essentially, the interesting situations distill into three cases, corresponding to the logics: $\{\exists,\wedge\}$, $\{\forall,\exists,\wedge\}$ and $\{\forall,\exists,\wedge,\vee\}$. Notice that $=$ is not in these logics a priori. For the first two logics it would not matter if we had added it as it can be propagated out. For the third, its absence is significant.

So, the first two logics correspond to the CSP and the \emph{quantified} CSP (QCSP), respectively. The classification for the former, over finite templates, was accomplished by \cite{BulatovFVConjecture,ZhukFVConjecture} with the resolution of the Feder-Vardi Conjecture (``CSP Dichotomy''), namely that all such problems $\CSP(\mathcal{B})$ are in P or are \NP-complete. The classification for the QCSP on finite templates is wide open, including exotic complexity classes such as DP-complete and $\Theta^\mathrm{P}_2$-complete \cite{ZhukM22}. The classification for the third logic, positive equality-free, over finite templates, was given in \cite{MadelaineM18}\footnote{The conference version of this paper was \cite{MadelaineM11}, so the result predates the CSP dichotomy.} as a tetrachotomy between P, \NP-complete, \coNP-complete and \Pspace-complete. No one would compare the resolution of this tetrachotomy to the classification for the CSP or QCSP, but some interesting mathematics was developed in its pursuit.

The \emph{algebraic approach} to CSPs dates back to the late 1990s with Jeavons's paper \cite{Jeavons98}. It relates universal algebraic objects called \emph{polymorphisms} to relations which they preserve when applied coordinatewise. Through a Galois connection, a classification is free to move between the relational objects (model or template) and algebraic objects (\emph{clones}). The algebraic approach was instrumental in the settling of the Feder-Vardi Conjecture. Such an algebraic approach has also been potent for QCSPs \cite{BBCJK,ZhukM22} and was developed in \cite{MadelaineM12} for positive equality-free logic. The algebraic objects are \emph{surjective hyper-operations} and play a central role in the tetrachotomy of \cite{MadelaineM18}, together with a new notion of core-ness. Some of the algebra has reappeared in the context of the promise problem in \cite{AsimiBB22} as well as in \cite{CarvalhoM22}. All references to algebra in this paper relate to universal algebra as just discussed.

In this paper we demonstrate that algebra and core-ness are not needed in the classification for positive equality-free logic.
Partly inspired by \cite{BodirskyHR13}, we give a direct proof\footnote{This proof, due to Kozik, bears the name \emph{direttissima} (a mountain-climbing term).} of the tetrachotomy. The proof from \cite{BodirskyHR13} concerns existential positive, $\{\exists,\wedge,\vee\}$, logic for which the corresponding model-checking problem gives a dichotomy between P and NP-complete across all templates (not just those that are finite). To extend the result to our logic, it seems as though one just has to deal additionally with the universal quantifier. To some extent this is true, but the complexity classification for existential positive logic has the following property. Assume $\mathrm{P}\neq \NP$, then the model-checking problem associated with $\mathcal{B}$ is $\NP$-hard iff there are existential positive definable non-empty relations $\phi_1$ and $\phi_2$, so that $\phi_1 \cap \phi_2=\emptyset$ (Lemma 5 in \cite{BodirskyHR13}). The generalisation of this statement does not hold when universal quantification is added: model-checking positive equality-free logic on $(\mathbb{Q},=)$ is \NP-complete, yet no relations $\phi_1$ and $\phi_2$ are positive equality-free definable on $(\mathbb{Q},=)$ so that $\phi_1$ and $\phi_2$ are non-empty, yet $\phi_1 \cap \phi_2 = \emptyset$.

In this paper, we give a simple proof of the tetrachotomy of model-checking positive equality-free logic for finite templates that does not involve algebra or $U$-$X$-cores. For arbitrary infinite templates, the method works to prove a dichotomy between NP-hard (under Turing reductions) and Logspace. The Turing reductions arise because hardness can be for co-NP as well as NP. 
What is remarkable is that the algebraic approach using surjective hyper-operations is problematic for infinite templates. It is not clear that the Galois connection of \cite{MadelaineM12} holds even for well-behaved infinite structures such as those that are homogeneous. Certainly the special surjective hyper-operations of \cite{MadelaineM12,MadelaineM18} that delineate Logspace, NP and co-NP (respectively named, $\forall\exists$-, $A$- and $E$-) no longer play that role, even if one were to define a corresponding $U$-$X$-core. It is interesting that our new method to prove NP membership (respectively, co-NP-membership) for finite templates fails already for $(\mathbb{Q},=)$ (respectively, $(\mathbb{Q},\neq)$).

Indeed, the modern, systematic study of infinite-domain CSPs began with a complexity classification for those templates which have a first-order definition in $(\mathbb{Q};=)$, usually known as \emph{equality (constraint) languages} \cite{BodirskyK08}. It continued, for example, with those templates with a first-order definition in $(\mathbb{Q};<)$, usually known as \emph{temporal (constraint) languages} \cite{BodirskyK10}, and those templates with a first-order definition in the Random Graph $(V,E)$ \cite{BodirskyP15}. For equality languages, the QCSP classification has only recently been fully accomplished \cite{BodirskyChenSICOMP,ZhukMW23}, with a trichotomy between Logspace, \NP-complete and \Pspace-complete.

We prove the tetrachotomy for model-checking positive equality-free logic on equality languages. The algorithm we use to drop complexity to \NP\ works by always evaluating a universal variable to a new element, distinct from any played for some variable earlier in the prefix order. This suggests possible algorithms for temporal languages, too; perhaps always playing a universal variable to a new element, strictly lower than any played for some variable earlier in the prefix order. Of course, there is also the possibility to play the new variable strictly higher. Indeed, we prove that these algorithms both work as well as one another, as we prove the tetrachotomy for model-checking positive equality-free logic on temporal languages. Alas, there are no more new tractable cases in the temporal languages (in Logspace, \NP\ or \coNP) than there were for the equality languages.

However, the case is different for first-order expansions of the Random Graph $(V,E)$. Let the binary relation $N$ hold on all distinct vertices which are not connected by $E$ (one can note that $(V;E)$ and $(V;N)$ are isomorphic). Here the problem associated with $(V, E)$ is in Logspace, and the algorithm is indeed to always evaluate a universal variable to a new element, that has an $N$-edge to all the previous elements. We finesse the tetrachotomy for model-checking positive equality-free logic on first-order expansions of the Random Graph $(V,E)$, using this algorithm, together with the dual one that always chooses for an existential variable a new element, that has an $E$-edge to all the previous elements.

We then go on to briefly consider the promise version of our problem, which has been introduced in \cite{AsimiBB22}. In this, the template is a pair $(\mathcal{A},\mathcal{B})$ of structures and the question involves answering yes to those positive equality-free inputs $\phi$ that are true on $\mathcal{A}$ and answering no to those that are not even true on $\mathcal{B}$. The promise is that the input $\phi$ is either true on $\mathcal{A}$ or false on $\mathcal{B}$; at least, we can answer anything if the promise is broken. We can also build our template so it is impossible for $\phi$ to be false on $\mathcal{B}$ but true on $\mathcal{A}$. In \cite{AsimiBB22}, considerable progress is made on classifying the complexity of this promise problem as $(\mathcal{A},\mathcal{B})$ vary over finite structures. We cannot solve the open cases they pose as open questions, but we can solve cases that are not covered in their paper. We show therefore how our methods can be used in this fashion.

The paper is organised as follows. After some preliminaries, we begin with general results for hardness in Section~\ref{sec:hardness}. We then address the straightforward case of finite structures in Section~\ref{sec:fin-class}. We move on to infinite structures in Section~\ref{sec:infinite-case}: first equality languages, then temporal languages, then first-order expansions of the Random Graph. We then conclude with some final remarks. Owing to reasons of space, some proofs are deferred to the appendix.  Our discussion of promise problems is deferred to the appendix in its entirety.


\section{Preliminaries}
\label{sec:prelims}

Let $[k]:=\{1,\ldots,k\}$. We use models, structures and templates interchangeably when talking about a model-checking problem. Though, for the promise version of the problem, the template becomes a pair of models or structures. If $\mathcal{B}$ is a structure then $B$ is its domain. All structures in this paper have finite relational signature.

A structure is \emph{homogeneous} if all isomorphisms between finite substructures can be extended to automorphisms. An infinite homogeneous structure is \emph{finitely bounded} if it can be described by finitely many forbidden finite induced substructures. 
\red{If a structure $\mathcal B$ is obtained from a structure $\mathcal A$ by removing relations, we say that ${\mathcal B}$ is a \emph{reduct} of ${\mathcal A}$ and that ${\mathcal A}$ is an \emph{expansion} of ${\mathcal B}$. 
A \emph{first-order expansion} of 
a structure $\mathcal B$ is an expansion of $\mathcal B$ by first-order definable relations. A \emph{first-order reduct} of $\mathcal B$ is a reduct of a first-order expansion of $\mathcal B$.} 
All infinite structures in this paper are reducts of finitely bounded homogeneous structures.

The Random Graph $(V;E)$ is the unique countable homogeneous graph that embeds all finite graphs. On the Random Graph we use vertices and elements interchangeably.

Let $\MC(\mathcal{B})$ be the problem to evaluate an input positive equality-free sentence on $\mathcal{B}$.
We always may assume that an instance of $\MC(\mathcal{B})$ is of the prenex form $$\forall x_1\exists y_1 \forall x_2 \exists y_2 \dots \forall x_{n} \exists y_{n} \, \theta,$$ where $\theta$ is quantifier-free (and positive equality-free), since if it is not it may readily be brought into an equivalent formula of this kind in logarithmic space. Then a solution is a sequence of (Skolem) functions 
$f_{1},\ldots, f_{n}$ such that 
$$(x_1,f(x_1),x_2,f_2(x_1,x_2),\ldots,x_{n},f_n(x_1,\ldots,x_n))$$ is a solution of $\Phi$ for all 
$x_{1},\ldots,x_{n}$ (i.e. $y_{i} = f_{i}(x_{1},\ldots,x_{i})$). This belies a (Hintikka) game semantics for the truth of an instance in which a player called Universal (male) plays the universal variables and a player called Existential (female) plays the existential variables, one after another, from the outside in. Universal aims to falsify the formula while Existential aims to satisfy it.
The Skolem functions above give a strategy for Existential. In our proofs we may occasionally revert to a game-theoretical parlance.

\begin{lemma}
Let $\mathcal{B}$ be finite, then $\MC(\mathcal{B})$ is in \Pspace.
\label{lem:Pspace}
\end{lemma}
\begin{proof}
Suppose $|B|=m$. If the input sentence has $n$ quantified variables, then cycle through all $m^n$ valuations of the variables (in exponential time). The data structure that keeps record of the current valuation is of size linear in $n$. The variables are addressed in prefix order with attention being paid to whether each is existential or universal once the cycle for that variable is complete. 
\end{proof}
This method of cycling through new possibilities is enough also for equality languages and temporal languages. For the Random Graph, the number of types grows too quickly, so we appeal to a more general algorithmic result. 
\begin{proposition}
Let $\mathcal{B}$ be a \red{first-order reduct} of a finitely bounded homogeneous structure. Then $\MC(\mathcal{B})$ is in \Pspace.
\label{prop:in-Pspace}
\end{proposition}
\begin{proof}
    Proposition 7(2) from \cite{Rydval18} proves that model-checking first-order sentences on a finitely bounded homogeneous structure is in \Pspace. We apply this after noting that an input instance can be rewritten over the signature of the underlying finitely bounded homogeneous structure in polynomial time.
\end{proof}
\subsection{The principle of duality}

If $\mathcal{B}$ is a structure, then define its dual, $\overline{\mathcal{B}}$, over the same domain $B$ but with $k$-ary relations $R$ replaced by $B^k\setminus R$. Note that $\mathcal{B}=\overline{\overline{\mathcal{B}}}$. Similarly, if $\phi$ is a positive equality-free sentence over $\Gamma$, then let $\overline{\phi}$ be the positive equality-free sentence obtained by rewriting $\neg \phi$ using de Morgan's laws to push negation innermost and then substituting negated atoms of $\mathcal{B}$ for (positive) atoms of $\overline{\mathcal{B}}$. Note that $\overline{\phi}$ is not equivalent to $\neg \phi$ (at least on the same structure). The following is clear from the construction.
\begin{lemma}[Principle of duality \cite{MadelaineM18}]
\red{A positive equality-free sentence} $\phi$ is a yes-instance of $\MC(\mathcal{B})$ iff $\overline{\phi}$ is a no-instance of $\MC(\overline{\mathcal{B}})$. It follows that:
\begin{itemize}
\item $\MC(\mathcal{B})$ is in Logspace iff $\MC(\overline{\mathcal{B}})$ is in Logspace;
\item $\MC(\mathcal{B})$ is in \NP\ iff $\MC(\overline{\mathcal{B}})$ is in \coNP;
\item $\MC(\mathcal{B})$ is \NP-complete iff $\MC(\overline{\mathcal{B}})$ is \coNP-complete.
\end{itemize}
\label{lem:duality}
\end{lemma}

\section{Hardness}
\label{sec:hardness}

We say that a positive equality-free sentence $\phi:=Q_1v_1\ldots Q_kv_k \, (\phi_1 \wedge \phi_2)$, where $\phi_1$ and $\phi_2$ are positive equality-free formulae, \emph{breaks $\wedge$ on $\mathcal{B}$} iff $\mathcal{B} \notmodels \phi$, though both $\mathcal{B} \models  Q_1v_1\ldots Q_kv_k \, \phi_1$ and $\mathcal{B} \models  Q_1v_1\ldots Q_kv_k \, \phi_2$. Note that it is in Existential's power to ensure either $\phi_1$ or $\phi_2$ is true and she can choose whichever she pleases, whereupon the other will become false.
Similarly, we say that $\psi:=Q'_1w_1\ldots Q'_\ell w_\ell \, (\psi_1 \vee \psi_2$) \emph{breaks $\vee$ on $\mathcal{B}$} iff $\mathcal{B} \models \psi$, though both $\mathcal{B} \notmodels  Q'_1w_1\ldots Q'_\ell w_\ell \, \psi_1$ and $\mathcal{B} \notmodels  Q'_1w_1\ldots Q'_kw_\ell \, \psi_2$. 
Note that it is in Universal's power to ensure either $\psi_1$ or $\psi_2$ is true and he can choose whichever he pleases, whereupon the other will become false. If $\phi$ breaks $\wedge$ on $\mathcal{B}$ then $Q_k$ is existential, but $Q_1,\ldots,Q_{k-1}$ can be arbitrary. Similarly, If $\psi$ breaks $\vee$ on $\mathcal{B}$ then $Q'_\ell$ is universal, but $Q'_1,\ldots,Q'_{\ell-1}$ can be arbitrary.

These definitions are inspired by Definition 2 from \cite{BodirskyHR13}. However, we don't have the key property of Lemma 6 from \cite{BodirskyHR13} (note that the arity of $\phi_1$ and $\phi_2$ from this lemma is strictly positive). To show this, consider the example $(\mathbb{Q};=)$. It is not possible to (positive equality-free) define disjoint and non-empty relations $\phi_1$ and $\phi_2$ over $(\mathbb{Q};=)$, yet the formula $\phi:=\forall x\forall y\exists z \ z=x \wedge z=y$ breaks $\wedge$ on $(\mathbb{Q};=)$. 

If there doesn't exist a $\phi$ so that $\phi$ breaks $\wedge$ on $\mathcal{B}$, then (on $\mathcal{B}$) $\exists$ commutes with both $\wedge$ and $\vee$. If there doesn't exist a $\psi$ so that $\psi$ breaks $\vee$ on $\mathcal{B}$, then (on $\mathcal{B}$) $\forall$ commutes with both $\wedge$ and $\vee$. 

The following lemma places no restriction on $\mathcal{B}$.
\begin{lemma}
Let $\mathcal{B}$ be a structure.
\begin{itemize}
\item If there exists $\phi$ so that $\phi$ breaks $\wedge$ on $\mathcal{B}$, then $\MC(\mathcal{B})$ is \NP-hard. 
\item If there exists $\psi$ so that $\psi$ breaks $\vee$ on $\mathcal{B}$, then $\MC(\mathcal{B})$ is \coNP-hard. 
\item If there exists $\phi$ and $\psi$ so that $\phi$ breaks $\wedge$ on $\mathcal{B}$ and $\psi$ breaks $\vee$ on $\mathcal{B}$,  then $\MC(\mathcal{B})$ is \Pspace-hard.
\end{itemize}
\label{lem:marcin-after-bodirsky-et-al}
\end{lemma}
\begin{proof}
Let us address the third case first, for we shall see that the first two cases are just specialisations of this. Let $\phi:=Q_1v_1\ldots Q_kv_k$ $(\phi_1 \wedge \phi_2)$ and let $\psi:=Q'_1w_1\ldots Q'_\ell w_\ell$ $(\psi_1 \vee \psi_2)$.


We will reduce from an instance $\theta$ of \emph{(monotone) Quantified $1$-in-$3$-satisfiability} (Q1-in-3SAT), known to be \Pspace-complete from \cite{Creignou}, to an instance $\theta'$ of $\mathrm{MC}(\mathcal{B})$. Let us recall that Q1-in-3-SAT takes a quantified conjunction of positive triples of variables, where the satisfying condition is that precisely one in each triple of variables is evaluated true after the quantifiers are played.

Note that variables of $\theta$ are partitioned into two types, existential and universal. We will handle the truth or falsity of these variables differently according to their type. Specifically, existential variables $x$ of $\theta$ will become a sequence of variables $v^x_1,\ldots,v^x_k$, with $\phi_1(v^x_1,\ldots,v^x_k)$ representing true and $\phi_2(v^x_1,\ldots,v^x_k)$ representing false. Universal variables $y$ of $\theta$ will become a sequence of variables $w^y_1,\ldots,w^y_\ell$, with $\psi_1(w^y_1,\ldots,w^y_\ell)$ representing true and $\psi_2(w^y_1,\ldots,w^y_\ell)$ representing false. 

We replace the quantification as we proceed inwards. Thus, $\exists x$ in $\theta$ becomes $Q_1v^x_1\ldots Q_kv^x_k$ in $\theta'$, and $\forall y$ in $\theta$ becomes $Q'_1w^y_1\ldots Q'_\ell w^y_\ell$ in $\theta'$.

It remains to explain how to represent clauses of the case $(p,q,r)$ and this depends on the form of the clauses, where the four cases are: all existential; all universal; one existential and two universal; two existential and one universal. The purely existential case involves adding to $\theta'$ that
\[
\begin{array}{ll} 
(\phi_1(v^p_1,\ldots,v^p_k) \wedge \phi_2(v^q_1,\ldots,v^q_k) \wedge \phi_2(v^r_1,\ldots,v^r_k)) & \vee \\ 
(\phi_2(v^p_1,\ldots,v^p_k) \wedge \phi_1(v^q_1,\ldots,v^q_k) \wedge \phi_2(v^r_1,\ldots,v^r_k)) & \vee \\ 
(\phi_2(v^p_1,\ldots,v^p_k) \wedge \phi_2(v^q_1,\ldots,v^q_k) \wedge \phi_1(v^r_1,\ldots,v^r_k)) & \\ 
\end{array}
\]
The purely universal case involves adding to $\theta'$ that 
\[
\begin{array}{ll} 
(\psi_1(w^p_1,\ldots,w^p_\ell) \wedge \psi_2(w^q_1,\ldots,w^q_\ell) \wedge \psi_2(w^r_1,\ldots,w^r_\ell)) & \vee \\ 
(\psi_2(w^p_1,\ldots,w^p_\ell) \wedge \psi_1(w^q_1,\ldots,w^q_\ell) \wedge \psi_2(w^r_1,\ldots,w^r_\ell)) & \vee \\ 
(\psi_2(w^p_1,\ldots,w^p_\ell) \wedge \psi_2(w^q_1,\ldots,w^q_\ell) \wedge \psi_1(w^r_1,\ldots,w^r_\ell)) \\ 
\end{array}
\]
The mixed cases work by mixing these two regimes. Suppose $p$ is existential and $q,r$ are universal. Then we add
\[
\begin{array}{ll} 
(\phi_1(v^p_1,\ldots,v^p_k) \wedge \psi_2(w^q_1,\ldots,w^q_\ell) \wedge \psi_2(w^r_1,\ldots,w^r_\ell)) & \vee \\ 
(\phi_2(v^p_1,\ldots,v^p_k) \wedge \psi_1(w^q_1,\ldots,w^q_\ell) \wedge \psi_2(w^r_1,\ldots,w^r_\ell)) & \vee \\ 
(\psi_2(v^p_1,\ldots,v^p_k) \wedge \psi_2(w^q_1,\ldots,w^q_\ell) \wedge \psi_1(w^r_1,\ldots,w^r_\ell)) \\ 
\end{array}
\]
Finally, if $p,q$ are existential and $r$ is universal, then we add
\[
\begin{array}{ll} 
(\phi_1(v^p_1,\ldots,v^p_k) \wedge \phi_2(v^q_1,\ldots,v^q_k) \wedge \psi_2(w^r_1,\ldots,w^r_\ell)) & \vee \\ 
(\phi_2(v^p_1,\ldots,v^p_k) \wedge \phi_1(v^q_1,\ldots,v^q_k) \wedge \psi_2(w^r_1,\ldots,w^r_\ell)) & \vee \\ 
(\psi_2(v^p_1,\ldots,v^p_k) \wedge \phi_2(v^q_1,\ldots,v^q_k) \wedge \psi_1(w^r_1,\ldots,w^r_\ell)) \\ 
\end{array}
\]
Let us argue that $\theta$ is a yes-instance of Q1-in-3SAT iff $\theta'$ is a yes-instance of $\MC(\mathcal{B})$.

(Forwards.) Existential mirrors her winning strategy for $\theta$ in $\theta'$ by considering all Universal plays of $\psi_1(w^y_1,\ldots,w^y_\ell)$ as true on $y$ while $\psi_2(w^y_1,\ldots,w^y_\ell)$ is false on $y$. She plays herself true variables $x$ as $v^x_1,\ldots,v^x_k$ so that $\phi_1(v^x_1,\ldots,v^x_k)$ holds and false variables as $v^x_1,\ldots,v^x_k$ so that $\phi_2(v^x_1,\ldots,v^x_k)$ holds. By construction it follows that $\theta'$ is true on $\mathcal{B}$.

(Backwards.) Suppose $\theta'$ is true on $\mathcal{B}$. Existential mirrors her winning strategy for $\theta'$ in $\theta$ by interpreting all Universal plays of $\psi_1(w^y_1,\ldots,w^y_\ell)$ as true on $y$ and $\psi_2(w^y_1,\ldots,w^y_\ell)$ as false on $y$. By construction, it follows that $\theta$ has a 1-in-3 satisfying assignment.

Hardness for \NP\ or \coNP\ simply uses only one of the two constructions for the types existential and universal. In these respective cases, the hardness follows from that for \emph{(monotone) $1$-in-$3$-satisfiability} \cite{Schaefer78} or its complement.
\end{proof}

Notwithstanding that this is a section on hardness, let us finish on the positive note of tractability in Logspace, where we make no assumptions about the structure other than that its signature is finite.
\begin{lemma}
Let $\mathcal{B}$ be any structure on a finite signature. If there does not exist $\phi$ so that $\phi$ breaks $\wedge$ on $\mathcal{B}$, and there does not exist $\psi$ so that $\psi$ breaks $\vee$ on $\mathcal{B}$, then $\MC(\mathcal{B})$ is in Logspace. 
\label{lem:easy}
\end{lemma}
\begin{proof}
Let $\theta$ be an input to $\MC(\mathcal{B})$ in the prenex form $\forall y_1 \exists x_1\ldots$ $\forall y_m \exists x_m \theta'$ where $\theta'$ is quantifier-free. Obtain $\theta''$ from $\theta'$ by substituting atoms $R(z_1,\ldots,z_k)$ by $\forall y_1 \exists x_1\ldots \forall y_k \exists x_k $ $R(z_1,\ldots,z_k)$, and indeed one may restrict the quantification to just the variables from $\{x_1,y_1,\ldots,x_m,y_m\}$ that appear in $\{z_1,\ldots,z_k\}$. Since $\mathcal{B}$ is finite signature, there is a finite number of such quantified atoms and we may assume there exists a finite table in which we can look up whether they evaluate to true or false. Now, by assumptions, both quantifiers commute with both conjunction and disjunction, so we may move all quantifiers inward towards the atoms, obtaining $\mathcal{B} \models \theta$ iff $\mathcal{B} \models \theta''$. This latter is a Boolean sentence evaluation problem which can be solved in Logspace \cite{Buss87}. 
\end{proof}
Note that in the last proof we do not specify a way to build the finite table. For us it is enough that it exists for each $\mathcal{B}$.

\section{The finite case}
\label{sec:fin-class}

\begin{lemma}
Let $\mathcal{B}$ be a finite structure such that there exists $\phi$ that breaks $\wedge$ on $\mathcal{B}$ but there is no $\psi$ that breaks $\vee$ on $\mathcal{B}$. Then $\MC(\mathcal{B})$ is in \NP.
\label{lem:NP-finite}
\end{lemma}
\begin{proof}
Consider a formula $\phi$ of the form $\forall x \exists y \, \phi'(x,y)$ which may contain additional free variables. Let $|B|=m$. Then $\phi$ is equivalent to $\exists y_1,\ldots,y_m \forall x \bigvee_{i \in [m]} \phi'(x,y_i)$. Since $\forall$ commutes with disjunction, this is equivalent to $\exists y_1,\ldots,y_m  \bigvee_{i \in [m]} \forall x \, \phi'(x,y_i).$
By symmetry, this is equivalent to $\exists y \forall x \, \phi'(x,y)$.

Let $\phi$ be an input to $\MC(\mathcal{B})$. If the innermost quantifier is $\forall$, then this commutes with both $\wedge$ and $\vee$ and can be pushed to the atomic level. If the innermost quantifier is $\exists$, then this can be swapped with some $\forall$ that is nearest moving outwards by the argument of the previous paragraph. This procedure can be iterated until the formula is purely existential modulo a language that is expanded by universally quantified atoms (the number of which is finite). 
\end{proof}
\begin{lemma}
Let $\mathcal{B}$ be a finite structure so that there exists $\psi$ that breaks $\vee$ on $\mathcal{B}$ but there is no $\phi$ that breaks $\wedge$ on $\mathcal{B}$. Then $\MC(\mathcal{B})$ is in \coNP.
\label{lem:coNP-finite}
\end{lemma}
\begin{proof}
This follows from Lemma~\ref{lem:duality}, when one notes that $\overline{\psi}$ breaks $\wedge$ on $\overline{\mathcal{B}}$, but there is no $\theta$ that breaks $\vee$ on $\overline{\mathcal{B}}$ (else $\overline{\theta}$ would break $\wedge$ on $\mathcal{B}$). 
\end{proof}
\begin{corollary}
Let $\mathcal{B}$ be finite. Then:
\begin{itemize}
\item If there does not exist $\phi$ that breaks $\wedge$ on $\mathcal{B}$, and there does not exist $\psi$ that breaks $\vee$ on $\mathcal{B}$, then $\MC(\mathcal{B})$ is in Logspace. 
\item If there exists $\phi$ that breaks $\wedge$ on $\mathcal{B}$, and there does not exist $\psi$ that breaks $\vee$ on $\mathcal{B}$, then $\MC(\mathcal{B})$ is \NP-complete. 
\item If there does not exist $\phi$ that breaks $\wedge$ on $\mathcal{B}$, and there exists $\psi$ that breaks $\vee$ on $\mathcal{B}$, then $\MC(\mathcal{B})$ is \coNP-complete. 
\item If there exists $\phi$ that breaks $\wedge$ on $\mathcal{B}$, and there exists $\psi$ that breaks $\vee$ on $\mathcal{B}$, then $\MC(\mathcal{B})$ is \Pspace-complete. 
\end{itemize}
\end{corollary}
\begin{proof}
The first case follows from Lemma~\ref{lem:easy}. Membership in the final case follows from  Lemma~\ref{lem:Pspace}. For the remaining cases, hardness follows from Lemma~\ref{lem:marcin-after-bodirsky-et-al} and membership follows from Lemmas~\ref{lem:NP-finite} and ~\ref{lem:coNP-finite}.
\end{proof}
The original proof of the previous result appears as Theorem 41 in \cite{MadelaineM18}, where the conditions for being in the respective classes are given by certain surjective hyper-endomorphisms. The new version of the result has a striking disadvantage. While the monoid of surjective hyper-endomorphisms is computable from a finite structure, it is not immediately clear how one computes whether there exists a certain formula that breaks $\wedge$ or $\vee$ on $\mathcal{B}$. To that extent, our discourse here is non-constructive as it does not solve the delineation of the classes (usually referred to as the meta-problem in the CSP community).

\section{The infinite case}
\label{sec:infinite-case}

The quantifier swapping method from Lemma~\ref{lem:NP-finite} fails already for $(\mathbb{Q};=)$. Note that $\MC(\mathbb{Q};=)$ is in \NP\ \cite{Kozen81}, and $\forall x \exists y \, x=y$ is true on $(\mathbb{Q};=)$, while $\exists y \forall x \, x=y$ is false. The majority of the paper is concerned with finding new methods to mitigate this.

\subsection{Equality languages}
\label{sec:equality}

Recall that an equality language is one that has a first-order definition in $(\mathbb{Q};=)$. Let us define the formula
$${\stackrel{\neq}{\forall} x_1} \exists \overline{y}_1 \ldots {\stackrel{\neq}{\forall} x_k} \exists \overline{y}_k \, \phi'(x_1,\overline{y}_1,\ldots,x_{k},\overline{y}_{k},z_1,\ldots,z_q)$$
by insisting that universal variables are always evaluated to an element distinct (different) from all outer quantified variables and free variables. Strictly, let us assume that the quantification is over all such possibilities. Though, for equality languages, there is only one such distinct (different) type up to automorphism. Let us dub the corresponding strategy for Universal as the \emph{all-different} strategy (noting though that Existential may repeat an element and there may be repetitions in the free variables). 

Let us note that quantifiers $\stackrel{\neq}{\forall}$ commute with themselves, viz ${\stackrel{\neq}{\forall}x}{\stackrel{\neq}{\forall}y}={\stackrel{\neq}{\forall}y}{\stackrel{\neq}{\forall}x}$, but not with $\forall$. For example, on the graph $K_2$, $\forall x{\stackrel{\neq}{\forall}y} E(x,y)$ is true, whereas $ {\stackrel{\neq}{\forall}y}\forall x E(x,y)$ is false.  
\begin{lemma}
Let $\mathcal{B}$ be an equality language. Suppose that the positive equality-free formula
$
\forall x \, \phi'(x,z_1,\ldots,z_q) 
$
is logically distinct from 
$
{\stackrel{\neq}{\forall} x} \, \phi'(x,z_1,\ldots,z_q)$.
Then there exists $\zeta$ such that $\zeta$ breaks $\vee$ on $\mathcal{B}$. Note that $\phi'$ is not necessarily quantifier-free.
\label{lem:neq-1}
\end{lemma}
\begin{proof}
We proceed by induction on $q$. Note that when the two are logically distinct, the former must be false at some point $(z_1,\ldots,$ $z_q)=(a_1,\ldots,a_q)$ while the latter is true. Suppose $q=1$, then $\forall x \, \phi'(x,z_1)$ is logically distinct from ${\stackrel{\neq}{\forall} x} \phi'(x,z_1)$. By assumption, $\theta(x,z_1):= \phi'(x,z_1)$ is logically equivalent to $x\neq z_1$. Now $\exists u \exists v \forall w$ $u\neq w \vee v \neq w$ breaks $\vee$ on $\mathcal{B}$.

Now suppose the statement of the lemma is true for $q=k$ and let us prove that it is true for $q=k+1$. We may assume that $\phi:=\forall x \phi'(x,z_1,\ldots,z_{k+1})$ is logically equivalent to ${\stackrel{\neq}{\forall} x} \, \phi'(x,z_1,\ldots,z_{k+1})$ whenever $z_1,\ldots,z_{k+1}$ are not all distinct (else we reduce to a previous case). 


Suppose $\phi(a_1,\ldots,a_{k+1})$ is false at some point such that $|\{a_1,\ldots,$ $a_{k+1}\}|<k+1$. Then we violate the inductive hypothesis; let us explain how. Choose the finest non-singleton partition under the equality relation for some $\{a_1,\ldots,a_{k+1}\}$ such that $\phi(a_1,\ldots,a_{k+1})$ is false. Note we need to forbid the extreme choice of the singletons as $\phi(a_1,\ldots,a_{k+1})$ is false when $|\{a_1,\ldots,a_{k+1}\}|=k+1$ by the assumption of the lemma.

W.l.o.g. assume that $a_{k+1}$ is a repeated element. Now replace $z_{k+1}$ with $x'$ and add universal quantification to this outermost. Then, at some point $(z_1,\ldots,z_k)=(a_1,\ldots,a_k)$:
\begin{equation}
\forall x' \forall x \, \phi'(x,a_1,\ldots,a_{k},x')
\label{equ:1:neq:1}
\end{equation}
is false, but 
\begin{equation}
{\stackrel{\neq}{\forall} x'} {\stackrel{\neq}{\forall} x} \, \phi'(x,a_1,\ldots,a_{k},x')
\label{equ:1:neq:2}
\end{equation}
is true. For the latter, there are two cases to consider. If the partition were a cover of the singleton (trivial) partition, \mbox{i.e.} precisely two elements are equivalent and all others are singletons, then truth follows from our original assumptions. Otherwise, the slightly finer partition born of separating $z_{k+1}$ from its equivalence class is such that $\phi$ itself is true here (with quantification $\forall x$, which implies the weaker ${\stackrel{\neq}{\forall} x}$).

By assumption (\mbox{ind.} hyp.) both
${\stackrel{\neq}{\forall} x'} \forall x \, \phi'(x,a_1,\ldots,a_{k},x')$ and 
${\stackrel{\neq}{\forall} x} \forall x' \, \phi'(x,a_1,\ldots,a_{k},x')$
are equivalent to the \eqref{equ:1:neq:1}. But now we violate the inductive hypothesis through either of these and the \eqref{equ:1:neq:2}. Thus $\phi(a_1,\ldots,a_{k+1})$ is true at every point such that $|\{a_1,\ldots,a_{k+1}\}|<k+1$.

Let $S\subseteq [k+1]$ so that, for $i \in S$, $\phi'(z_i,z_1,\ldots,z_{k+1})$ is false (note that $S$ is non-empty by assumption). Then $$\theta(x,z_1,\ldots,z_{k+1})=\left( \bigwedge_{i\neq j \in [k+1]} z_i \neq z_j\right) \rightarrow \left(\bigwedge_{i \in S} x \neq z_i \right).$$ Now, we universally quantify over all $z_i$ such that $i \in [k+1] \setminus S$ and rename indices in the $z$-variables to obtain, for some $1\leq r$:
\[ \bigvee_{i\neq j \in [r]} z_i = z_j  \vee \bigwedge_{i \in [r]} x \neq z_i.\]
There are now several ways to conclude the argument, let us choose one. Note that $z_1=z_2$ is definable by universally quantifying all variables other than $z_1$ and $z_2$. Now the formula 
\[ \forall z_1,\ldots,z_r,x \left( \bigvee_{i\neq j \in [r]} z_i = z_j  \vee \bigwedge_{i \in [r]} x \neq z_i \right) \vee \left( \bigvee_{i \in r} z_i=x\right)\]
breaks $\vee$ on $\mathcal{B}$.
\end{proof}
The proof of the following lemma is deferred to the appendix.
\begin{lemma}
Let $\mathcal{B}$ be an equality language. Suppose that the positive equality-free formula $$\forall x_1 \exists \overline{y}_1 \ldots \forall x_k \exists \overline{y}_k \, \phi'(x_1,\overline{y}_1,\ldots,x_k,\overline{y}_k,z_1,\ldots,z_q)$$ is logically distinct from $${\stackrel{\neq}{\forall} x_1} \exists \overline{y}_1 \ldots {\stackrel{\neq}{\forall} x_k} \exists \overline{y}_k \, \phi'(x_1,\overline{y}_1,\ldots,x_k,\overline{y}_k,z_1,\ldots,z_q).$$ Then there exist $\zeta$ so that $\zeta$ breaks $\vee$ on $\mathcal{B}$.
\label{lem:equ-appendix}
\end{lemma}


\begin{corollary}
Let $\mathcal{B}$ be an equality language. The all-different strategy is optimal for Universal iff $\vee$ does not break on $\mathcal{B}$.
\end{corollary}

\begin{lemma}
    Let $\mathcal{B}$ be an equality language. If the all-different strategy is optimal for Universal, then $\MC(\mathcal{B})$ is in \NP.
\end{lemma}
\begin{proof}
    When some elements have been played by Universal and Existential, there is a unique up to isomorphism new element that is not equal to all those played before (this is provided by homogeneity) and Universal always may be assumed to play this. Existential, meanwhile, plays either an element such as this, or some element that has gone before, and this guessing alone pushes the complexity into \NP.
\end{proof}

\begin{corollary}
Let $\mathcal{B}$ be an equality language. Either $\vee$ does not break on $\mathcal{B}$ and $\MC(\mathcal{B})$ is in \NP, or $\MC(\mathcal{B})$ is \coNP-hard.
\label{cor:equality-NP}
\end{corollary}

\begin{theorem}
Let $\mathcal{B}$ be an equality language. Either $\MC(\mathcal{B})$ is in L, is \NP-complete, is \coNP-complete or is \Pspace-complete.
\end{theorem}
\begin{proof}
If there does not exist $\phi$ so that $\phi$ breaks $\wedge$ on $\mathcal{B}$, and there does not exist $\psi$ so that $\psi$ breaks $\vee$ on $\mathcal{B}$, then $\MC(\mathcal{B})$ is in Logspace by Lemma~\ref{lem:easy}.

If there does exist $\phi$ so that $\phi$ breaks $\wedge$ on $\mathcal{B}$, but there does not exist $\psi$ so that $\psi$ breaks $\vee$ on $\mathcal{B}$, then $\MC(\mathcal{B})$ is in \NP\ by Corollary~\ref{cor:equality-NP} and is \NP-hard by Lemma~\ref{lem:marcin-after-bodirsky-et-al}. The dual case of \coNP-completeness follows from the principle of duality (Lemma~\ref{lem:duality}). Finally, if there exists $\phi$ so that $\phi$ breaks $\wedge$ on $\mathcal{B}$, and there exists $\psi$ so that $\psi$ breaks $\vee$ on $\mathcal{B}$, then $\MC(\mathcal{B})$ is \Pspace-hard by Lemma~\ref{lem:marcin-after-bodirsky-et-al} and in \Pspace\ by Proposition~\ref{prop:in-Pspace}.
\end{proof}

\subsection{Temporal languages}
\red{In the terminology of Section~\ref{sec:prelims}, 
a temporal language is a first-order reduct of $(\mathbb{Q};<)$.} This entire section is deferred to the appendix as it proceeds similarly to the case of equality languages.

\begin{theorem}
Let $\mathcal{B}$ be a first-order reduct of $(\mathbb{Q};<)$. Either $\MC(\mathcal{B})$ is in L, is \NP-complete, is \coNP-complete, or is \Pspace-complete.
\label{thm:temp-main}
\end{theorem}

\subsection{The Random Graph}

Throughout this section, let $\mathcal{B}$ be a first-order educt of the Random Graph $(V;E)$.
Let us define the formula ($E$-hat)
$${\stackrel{E}{\forall} x_1} \exists \overline{y}_1 \ldots {\stackrel{E}{\forall} x_k} \exists \overline{y}_k \, \phi'(x_1,\overline{y}_1,\ldots,x_{k},\overline{y}_{k},z_1,\ldots,z_q)$$
by insisting that universal variables are always evaluated to an element distinct from all outer quantified variables and free variables such that there is an $E$-edge from all the elements that have taken part in the evaluation to this new element. Let us define the like sentence but with ${\stackrel{N}{\forall}}$ ($N$-hat) dually, \mbox{i.e.}, with $N$-edges. Strictly, let us assume that the quantification is over all such possibilities. Though, for the Random Graph, there is only one such distinct type up to automorphism, and furthermore this type always exists. Let us dub the corresponding strategy for Universal as the \emph{all-$E$} and \emph{all-$N$} strategies, respectively. Let us similarly define quantifiers of the form ${\stackrel{E}{\exists}}$ and ${\stackrel{N}{\exists}}$ and the corresponding \emph{exists-$E$} and \emph{exists-$N$} strategies.

Let us now assume that $E$ is always present in our reduct $\mathcal{B}$, \mbox{i.e.}, $\mathcal{B}$ is a first-order expansion of $(V;\red{E})$. It will turn out that we no longer need to consider the quantifiers ${\stackrel{N}{\forall}}$ and ${\stackrel{N}{\exists}}$.

\begin{lemma}
Suppose that the ${\stackrel{E}{\exists}}$ strategy is not optimal for Existential on $\mathcal{B}$. Then there is some positive equality-free \red{formula} $\psi$ over $\mathcal{B}$ so that $\psi$ breaks $\wedge$ on $\mathcal{B}$.
\label{lem:random-graph-exists}
\end{lemma}
\begin{proof}
    Let $\phi:=\forall x_1 \exists y_1 \ldots \forall x_k \exists y_k \phi'(x_1,y_1,\ldots,x_k,y_k)$ be true on $\mathcal{B}$ such that $\forall x_1 {\stackrel{E}{\exists}} y_1 \ldots$ $\forall x_k {\stackrel{E}{\exists}} y_k \phi'(x_1,y_1,\ldots,x_k,y_k)$ is false.
Consider
\[ 
\forall x_1 \exists y_1 \ldots \forall x_k \exists y_k \ \phi'(x_1,y_1,\ldots,x_k,y_k) \wedge
 \bigwedge_{i \in [k]} \bigwedge_{\begin{array}{c}\mbox{$v$ comes before} \\[-4pt]\mbox{$y_i$ in prefix}\end{array}} E(v,y_i).
\]
By assumption this is false, but $\forall x_1 \exists y_1 \ldots \forall x_k \exists y_k \phi'(x_1,y_1,\ldots,x_k,y_k)$ is true and 
\[ \forall x_1 \exists y_1 \ldots \forall x_k \exists y_k \bigwedge_{i \in [k]} \bigwedge_{\begin{array}{c}\mbox{$v$ comes before} \\[-4pt]\mbox{$y_i$ in prefix}\end{array}} E(v,y_i)\]
is true. Therefore, we have broken $\wedge$ on $\mathcal{B}$.
\end{proof}
The following lemma is not completely dual to Lemma~\ref{lem:random-graph-exists} as we still consider a first-order expansion of $(V;E)$.

\begin{lemma}
Suppose that the ${\stackrel{N}{\forall}}$ strategy is not optimal for Universal on $\mathcal{B}$. Then there is some positive equality-free $\psi$ over $\mathcal{B}$ so that $\psi$ breaks $\vee$ on $\mathcal{B}$.
\label{lem:random-graph-forall}
\end{lemma}
\begin{proof}
    Let $\phi:=\forall x_1 \exists y_1 \ldots \forall x_k \exists y_k \phi'(x_1,y_1,\ldots,x_k,y_k)$ be false on $\mathcal{B}$ such that ${\stackrel{N}{\forall}} x_1 \exists y_1 \ldots$ ${\stackrel{N}{\forall}} x_k \exists y_k \phi'(x_1,y_1,\ldots,x_k,y_k)$ is true.
Consider
\[
\exists w_1 \forall x_1 \exists y_1 \ldots \exists w_k \forall x_k \exists y_k \phi'(x_1,y_1,\ldots,x_k,y_k) \vee
\bigvee_{i \in [k-1]} \bigvee_{\begin{array}{c}\mbox{$v$ comes before} \\[-4pt]\mbox{$y_i$ in prefix}\end{array}} E(v,y_i),
\]
where we have introduced new existential variables $w_i$ immediately preceding each universal variable $x_i$. 
By assumption this is true, let us explain why. If Universal ever deviates from the all-$N$ strategy, it is because he plays an element that has already been played, or a new element that has an $E$-edge to some previous element. It is easy to see we cover the latter case in the big disjunction, but we also cover the former because Existential chooses the $w_i$ to be in an $E$-clique with one another and with $E$-edges to everything that has gone before.

However, $\exists w_1 \forall x_1 \exists y_1 \ldots \exists w_k \forall x_k \exists y_k \phi'(x_1,y_1,\ldots,x_k,y_k)$ is false and 
\[ \exists w_1 \forall x_1 \exists y_1 \ldots \exists w_k \forall x_k \exists y_k \bigvee_{i \in [k-1]} \bigvee_{\begin{array}{c}\mbox{$v$ comes before} \\[-4pt]\mbox{$y_i$ in prefix}\end{array}} E(v,y_i)\]
is false (Universal plays $x_k$ so that it has an $N$-edge to everything that has gone before). Therefore, we have broken $\vee$ on $\mathcal{B}$.
\end{proof}
\noindent Note that the final variable $y_k$ played no role in the previous proof. We left it there only because it was there in Lemma~\ref{lem:random-graph-exists}.
The following results now follow just as in the temporal and equality cases.
\begin{lemma}
    Let $\mathcal{B}$ be a first-order expansion of $(V;E)$. If the all-$E$ strategy is optimal for Universal, then $\MC(\mathcal{B})$ is in \NP. If the exists-$E$ strategy is optimal for Universal, then $\MC(\mathcal{B})$ is in \coNP.
\end{lemma}


\begin{corollary}
Let $\mathcal{B}$ be a first-order expansion of $(V;E)$. Either $\vee$ does not break on $\mathcal{B}$ and $\MC(\mathcal{B})$ is in \NP, or $\MC(\mathcal{B})$ is \coNP-hard. Either $\wedge$ does not break on $\mathcal{B}$ and $\MC(\mathcal{B})$ is in \coNP, or $\MC(\mathcal{B})$ is \NP-hard.
\label{cor:random-graph-NP}
\end{corollary}

\begin{theorem}
Let $\mathcal{B}$ be a first-order expansion of $(V;E)$. Either $\MC(\mathcal{B})$ is in Logspace, is \NP-complete, is \coNP-complete or is \Pspace-complete.
\end{theorem}
The all-$N$ strategy is optimal for Universal on $(V,E)$ and the exists-$E$ strategy is optimal for Existential. Therefore, $\MC(V,E)$ is in Logspace and we can see that there are new tractable cases, for the Random Graph, compared to equality languages (where there were no such new tractable cases for temporal languages).

\section{Final remarks}
It is mildly lamentable that we did not complete the complexity classification for all \red{first-order} reducts of the Random Graph.
%
It seems we need some new methods. For example, consider the \red{first-order reduct $\mathcal B$ of the Random Graph which 
contains a single relation of arity four
which contains a tuple $(a_1,a_2,a_3,a_4)$ if either 
 $|\{a_1,a_2,a_3,a_4\}|<4$,
 or $|\{a_1,a_2,a_3,a_4\}|=4$
 and $\{a_1,a_2,a_3,a_4\}$ induces a triangle in $E$-edges extended by a new vertex to which the three existing vertices are joined by $N$-edges.} 
 At present, we do not know how to handle this case, so as to prove that $\vee$ is broken on $\mathcal{B}$.


Let us comment on another case which we can solve. \red{Let $S \subseteq V^{\ell}$, for $\ell \geq 3$, be the relation that consists of precisely those tuples $(a_1,\ldots,a_\ell)$ where $|\{a_1,\ldots,a_\ell\}|<\ell$ 
or where $\{a_1,\ldots,a_\ell\}$ induces a clique of size $l$ in $E$ or in $N$. The relation $E$ does not have a first-order definition in $(V;S)$, because any isomorphism between $(V;E)$ and $(V;N)$ (clearly, there are such isomorphisms)
is an automorphism of $(V;S)$, but does not preserve $E$. 
For the structure $(V;S)$} we can appeal to Ramsey's theorem for the breaking of $\vee$,  but we need $R(\ell,\ell)$ variables where $R$ is the Ramsey function. \red{Let $\phi_\ell$ be the disjunction of $S(x_1,\ldots,x_\ell)$ over all size $\ell$ subsets $\{x_1,\ldots,x_\ell\}$ of $\{v_1,\ldots,v_{R(\ell,\ell)}\}$.
The universal quantification of $\phi_\ell$ is true by Ramsey's theorem, yet} the universal quantification of each individual \red{disjunct} is false. It follows that some single split of the big disjunction gives a single disjunction that breaks $\vee$ (this argument will appear in Lemma~\ref{lem:little-temp-1}). 

Finally, let us comment on the algebraic method. We never defined the special $\forall\exists$-, A- and E-surjective hyperoperations (shops) that played so central a role in \cite{MadelaineM18}. However, let us do so, at least for the first. A function $f \colon B\rightarrow \mathcal{P}(B)$, where $\mathcal{P}(B)$ is the power set of $\mathcal{B}$, is a \emph{shop} iff $\forall x \exists y \, y \in f(x)$ and $\forall y \exists x \, y \in f(x)$. Let $\mathcal{B}$ be a graph. Then we say that $f$ is a surjective hyper-endomorphism (she) of $\mathcal{B}$ iff, for all $x,y$: $E(x,y)$ implies $\forall x',y'$ $x' \in f(x), y' \in f(y)$ implies $E(x',y')$. Now, $f$ is a $\forall \exists$-shop iff $\exists x \forall y \, y \in f(x)$ and $\exists y \forall x \, y \in f(x)$. In the classification of Theorem 41 in \cite{MadelaineM18}, for finite $\mathcal{B}$, $\MC(\mathcal{B})$ is in Logspace iff $\mathcal{B}$ has a $\forall \exists$-she. Let us note that the Random Graph $(V;E)$ does not have a $\forall \exists$-she: suppose $f$ were a $\forall \exists$-she. Pick $y'$ so that for all $x$ we have $y' \in f(x)$. This is a contradiction as some edge exists in the Random Graph but $E(y',y')$ does not hold. Thus, we achieved more by leaving the algebraic method for this problem, because we have been able to cover certain infinite templates. 



\section{Appendix}

\noindent \textbf{Lemma~\ref{lem:equ-appendix}.}
Let $\mathcal{B}$ be an equality language. Suppose that the positive equality-free formula $$\forall x_1 \exists \overline{y}_1 \ldots \forall x_k \exists \overline{y}_k \, \phi'(x_1,\overline{y}_1,\ldots,x_k,\overline{y}_k,z_1,\ldots,z_q)$$ is logically distinct from $${\stackrel{\neq}{\forall} x_1} \exists \overline{y}_1 \ldots {\stackrel{\neq}{\forall} x_k} \exists \overline{y}_k \, \phi'(x_1,\overline{y}_1,\ldots,x_k,\overline{y}_k,z_1,\ldots,z_q).$$ Then there exist $\zeta$ so that $\zeta$ breaks $\vee$ on $\mathcal{B}$.

\begin{proof}
We prove this by induction on $k$. The base case $k=1$ is given by the previous lemma (when one notes that the innermost existential quantifiers may be absorbed into $\phi'$). Suppose that the statement is true for $k=m$. Let us consider the case $k=m+1$. By assumption there exists $a_1,\ldots,a_q$ so that
\begin{equation}
\forall x_1 \exists \overline{y}_1 \ldots \forall x_m \exists \overline{y}_m \forall x_{m+1} \exists \overline{y}_{m+1} \, \phi'(x_1,\overline{y}_1,\ldots,x_{m+1},\overline{y}_{m+1},a_1,\ldots,a_q)
\label{equ:neq-proof-1}
\end{equation}
is false where 
\begin{equation}
\stackrel{\neq}{\forall} x_1 \exists \overline{y}_1 \ldots \stackrel{\neq}{\forall} x_m \exists \overline{y}_m \stackrel{\neq}{\forall} x_{m+1} \exists \overline{y}_{m+1} \, \phi'(x_1,\overline{y}_1,\ldots,x_{m+1},\overline{y}_{m+1},a_1,\ldots,a_q)
\label{equ:neq-proof-2}
\end{equation}
is true. By the inductive hypothesis, \eqref{equ:neq-proof-1} is logically equivalent to 
\begin{equation}
\stackrel{\neq}{\forall} x_1 \exists \overline{y}_1 \ldots \stackrel{\neq}{\forall} x_m \exists \overline{y}_m \forall x_{m+1} \exists \overline{y}_{m+1} \, \phi'(x_1,\overline{y}_1,\ldots,x_{m+1},\overline{y}_{m+1},a_1,\ldots,a_q).
\label{equ:neq-proof-3}
\end{equation}
Thus, \eqref{equ:neq-proof-2} is true and \eqref{equ:neq-proof-3} is false. Now there must exist some assignment  $b_1,\overline{c}_1,\ldots,b_{m},\overline{c}_{m}$ to
$x_1,\overline{y}_1,\ldots,x_{m},\overline{y}_{m}$ so that
\begin{equation}
\stackrel{\neq}{\forall} x_{m+1} \exists \overline{y}_{m+1} \, \phi'(b_1,\overline{c}_1,\ldots,b_{m},\overline{c}_{m},x_{m+1},\overline{y}_{m+1},a_1,\ldots,a_q)
\label{equ:neq-proof-4}
\end{equation}
is true but
\begin{equation}
\forall x_{m+1} \exists \overline{y}_{m+1} \, \phi'(b_1,\overline{c}_1,\ldots,b_{m},\overline{c}_{m},x_{m+1},\overline{y}_{m+1},a_1,\ldots,a_q).
\label{equ:neq-proof-5}
\end{equation}
is false, and we violate the inductive hypothesis.
\end{proof}

\subsection*{Temporal languages}
Let us ponder what kind of algorithm we might use on first-order reducts of $(\mathbb{Q};<)$. We might consider quantifiers of the form ${\stackrel{<}{\forall} x}$, in which we consider only elements strictly below those that have already appeared, or ${\stackrel{>}{\forall} x}$, in which we consider only elements strictly above those that have already appeared. Then there would be the corresponding guarded existential quantification as well.

For a relation $\psi(v_1,\ldots,v_n)$ in precisely $n$ free variables, let $\bigvee_{\pi} \psi$ be a shorthand for  $\bigvee_{\pi \in S_n} \psi(v_{\pi(1)},\ldots,v_{\pi(n)})$, where $S_n$ is the symmetric group on $n$ elements.

\begin{lemma}
Let $\psi(v_1,\ldots,v_n)$ be a first-order relation over $(\mathbb{Q};<)$. Then $$\bigvee_{\pi \in S_n} \psi(v_{\pi(1)},\ldots,v_{\pi(n)})$$ is a first-order relation over $(\mathbb{Q};=)$.
\label{lem:temporal-disjunction}
\end{lemma}
\begin{proof}
In the following, we refer to the application of some automorphism of $(\mathbb{Q};<)$ as a re-scaling.
It suffices to prove that $S(v_1,\ldots,v_n)=\bigvee_{\pi \in S_n} \psi(v_{\pi(1)},\ldots,v_{\pi(n)})$ is closed under all permutations $\sigma$ of $\mathbb{Q}$. Consider $(a_1,\ldots,a_n) \in S$. Let us assume w.l.o.g. by re-ordering the co-ordinates and some rescaling and removing duplicates that $a_1< \cdots <a_n \in \{1,\ldots,n\}$ which implies that $a_1=1,\ldots,a_n=n$. Now, $(\sigma(a_1),\ldots,\sigma(a_n))$ is a rescaling of $(a_{\sigma(1)},\ldots,a_{\sigma(n)})$ and the result follows.
\end{proof}

\begin{lemma}
Let $\mathcal{B}$ be a first-order reduct of $(\mathbb{Q};<)$ with $\phi_1,\ldots,\phi_m$ positive equality-free formulas over $\mathcal{B}$ with free variables all among $v_1,\ldots,v_n$. Suppose that  
$\forall v_1,\ldots,v_n \, (\phi_1 \vee \ldots \vee \phi_m)$ is true on $\mathcal{B}$, but $\forall v_1,\ldots,v_n \, \phi_i$ is false on $\mathcal{B}$, for all $i \in [m]$. Then there exists some positive equality-free definable $\psi$ over $\mathcal{B}$ so that $\psi$ breaks $\vee$ in $\mathcal{B}$. 
\label{lem:little-temp-1}
\end{lemma}
\begin{proof}
Choose $k \in [m]$ minimally so that $\forall v_1,\ldots,v_n \, (\phi_1 \vee \ldots \vee \phi_{k+1})$ is true on $\mathcal{B}$. By assumption $m > k \geq 1$. Then let $\psi$ be $\forall v_1,\ldots,v_n \, (\phi_1 \vee \ldots \vee \phi_{k}) \vee \phi_{k+1}$ and note that this breaks $\vee$ by definition.
    \end{proof}

\begin{lemma}
Let $\mathcal{B}$ be a first-order reduct of $(\mathbb{Q};<)$ with $\phi$ a positive equality-free formula of $\mathcal{B}$ in precisely the free variables $v_1,\ldots,v_n$. Suppose that $\forall v_1, \ldots, \forall v_n \bigvee_{\pi} \phi$ is false but $\bigvee_{\pi} \phi(a_1,\ldots,a_n)$ holds at all points so that $|\{a_1,\ldots,a_{n}\}|=n$. Then at least one of $=$ and $\neq$ are positive equality-free definable on $\mathcal{B}$.
\label{lem:little-temp-2}
\end{lemma}
\begin{proof}
We split on whether $\phi(x,\ldots,x)$ is true (by homogeneity it is true at one point iff it is true everywhere).

($\phi(x,\ldots,x)$ is true.) Let us consider some coarsest partition $P=(P_1,\ldots,P_r)$ of $\{v_1,\ldots,v_n\}$ such that $\phi(v_1,\ldots,v_n)$ becomes false when we identify the elements of each class. Let us create new variables $u_1,\ldots,u_r$ for the classes. Then, by definition, $\phi(u_1,\ldots,u_r)$ defines $\bigvee_{i\neq j \in [r]} u_i=u_j$. If we now universally quantify all variables other than $u_1$ and $u_2$ we will define $u_1=u_2$.

($\phi(x,\ldots,x)$ is false.)  Let us consider some coarsest partition $P=(P_1,\ldots,P_r)$ of $\{v_1,\ldots,v_n\}$ such that when we identify the elements of each class $\phi(v_1,\ldots,v_n)$ becomes true. Let us create new variables $u_1,\ldots,u_r$ for the classes. Then, by definition, $\phi(u_1,\ldots,u_r)$ defines $\bigwedge_{i\neq j \in [r]} u_i\neq u_j$. If we now universally quantify all variables other than $u_1$ and $u_2$ we will define $u_1\neq u_2$.
\end{proof}

\begin{lemma}
Let $\mathcal{B}$ be a first-order reduct of $(\mathbb{Q};<)$. Suppose that the positive equality-free formula
\[
\forall x \, \phi'(x,z_1,\ldots,z_q) 
\]
is logically distinct from 
\[
{\stackrel{<}{\forall} x} \, \phi'(x,z_1,\ldots,z_q).
\]
Then there exists $\zeta$ such that $\zeta$ breaks $\vee$ on $\mathcal{B}$. Note that $\phi'$ is not necessarily quantifier-free.
\label{lem:temp-1}
\end{lemma}
\begin{proof}
We proceed by induction on $q$. Note that when the two are logically distinct, the former must be false at some point $(z_1,\ldots,$ $z_q)=(a_1,\ldots,a_q)$ while the latter is true. Suppose $q=1$, then $\forall x \, \phi'(x,z_1)$ is logically distinct from ${\stackrel{<}{\forall} x} \phi'(x,z_1)$. Consider $$\theta(x,z_1):= \phi'(x,z_1).$$ By assumption, this is logically equivalent to one of $<,\leq,\neq$. Now $\exists u \exists v \forall w \, u\neq w \vee v \neq w$ breaks $\vee$ on $\mathcal{B}$. For $<$ and $\leq$ we can produce something similar.

Now suppose it is true for $q=k$ and let us prove that it is true for $q=k+1$. We may assume that $\phi:=\forall x \phi'(x,z_1,\ldots,z_{k+1})$ is logically equivalent to ${\stackrel{<}{\forall} x} \, \phi'(x,z_1,\ldots,z_{k+1})$ whenever $z_1,\ldots,z_{k+1}$ are not all distinct (else we reduce to a previous case). 

Suppose that $\forall x,z_1,\ldots,x_n \bigvee_{\pi} \phi'(x,z_1,\ldots,z_{k+1})$ is true. Then we are in the situation of Lemma~\ref{lem:little-temp-1}. 

Thus it must be false. Yet, we know from our assumptions that $\bigvee_{\pi} \phi'(x,z_1,\ldots,z_{k+1})$ is true at every point in which the variables are evaluated as distinct elements. Now we are in the situation of Lemma~\ref{lem:little-temp-2}. If $\neq$ is definable then we know that we break $\vee$. Thus, $=$ must be definable, Now we consider
\[
\begin{array}{r}
\forall x,z_1,\ldots,x_n \bigvee_{\pi} \phi'(x,z_1,\ldots,z_{k+1}) \vee x=z_1 \vee \ldots \vee x=z_{k+1} \vee \\
\bigvee_{i \neq j \in [k+1]} z_i=z_j
\end{array}
\]
and note that we are again in the situation of Lemma~\ref{lem:little-temp-1}.
\end{proof}
Now we follow a sequence of proofs that proceed just as in the case of equality languages.
\begin{corollary}
Let $\mathcal{B}$ be a first-order reduct of $(\mathbb{Q};<)$. The all-different strategy is optimal for Universal iff $\vee$ does not break on $\mathcal{B}$.
\end{corollary}

\begin{lemma}
    Let $\mathcal{B}$ be a first-order reduct of $(\mathbb{Q};<)$. If the all-different strategy is optimal for Universal, then $\MC(\mathcal{B})$ is in \NP.
\end{lemma}
\begin{proof}
    When some elements have been played by Universal and Existential, there is a unique up to isomorphism new element that is strictly less than all those played before (this is provided by homogeneity) and Universal always may be assumed to play this. Existential, meanwhile, plays either some element that has gone before; or one in between, or strictly less than, or strictly greater than elements that have gone before. This guessing alone pushes the complexity into \NP.
\end{proof}

\begin{corollary}
Let $\mathcal{B}$ be a first-order educt of $(\mathbb{Q};<)$. Either $\vee$ does not break on $\mathcal{B}$ and $\MC(\mathcal{B})$ is in \NP, or $\MC(\mathcal{B})$ is \coNP-hard.
\label{cor:temporal-NP}
\end{corollary}

\

\noindent \textbf{Theorem~\ref{thm:temp-main}}.
Let $\mathcal{B}$ be a first-order reduct of $(\mathbb{Q};<)$. Either $\MC(\mathcal{B})$ is in L, is \NP-complete, is \coNP-complete or is \Pspace-complete.

\

The reader will have noticed that there was nothing special in our discourse to ${\stackrel{<}{\forall} x}$ that could not also have been accomplished with ${\stackrel{>}{\forall} x}$. Could it already have been accomplished by ${\stackrel{\neq}{\forall} x}$? The reader will probably not be surprised by the following answer, already prefigured in Lemma~\ref{lem:temporal-disjunction}.
\begin{proposition}
Let $\mathcal{B}$ be a first-order reduct of $(\mathbb{Q};<)$ in which an optimal algorithm for Universal is to always choose an element that is smaller than those that have been previously played. Then $\mathcal{B}$ is a first-order reduct of $(\mathbb{Q};=)$.
\end{proposition}
\begin{proof}
Let $R$ be a relation of $\mathcal{B}$ over $\mathbb{Q}$. We will prove that $R$ is invariant under all permutations of $\mathbb{Q}$. 

 Consider $(a_1,\ldots,a_{n}) \in R$. Let us assume \mbox{w.l.o.g.} by re-ordering the co-ordinates and some rescaling and removing duplicates that $a_1< \cdots <a_n \in \{1,\ldots,n\}$. Consider $\forall v_1 R(v_1,a_2,\ldots,a_n)$. This is true iff ${\stackrel{<}{\forall} v_1} R(v_1,a_2,\ldots,a_n)$ is true. The latter is true, so therefore so is the former. Thus we may reassign the first element to any element, say between $n+1$ and $2n$. By proceeding in this way, left to right, we may reassign all of the numbers $a_1,\ldots,a_n$ arbitrarily, and the result follows (potentially after some translation and rescaling).
\end{proof}
In the previous proof, the equality language is of a very special form -- it is positively definable in $(\mathbb{Q};=)$. That is, there can be no instance of $\neq$. This explains why we can take $a_i\neq a_j$ yet potentially move them to the same element (\mbox{i.e.} violating $\neq$).

\section*{Promise Problems}


The methods of Lemma~\ref{lem:marcin-after-bodirsky-et-al} work equally well for the promise version of the problem, $\PMC(\mathcal{A},\mathcal{B})$, as introduced in \cite{AsimiBB22}. Here, we take an input positive equality-free formula $\phi$ and we must respond with yes, if it true on $\mathcal{A}$, and no, if it is false on $\mathcal{B}$. We will choose $\mathcal{A}$ and $\mathcal{B}$ so that any positive equality-free input that is true on $\mathcal{A}$ is true on $\mathcal{B}$. In the event that $\phi$ is false on $\mathcal{A}$ but true on $\mathcal{B}$, it does not matter what we answer.

So long as there exists a single $\phi$ that breaks $\wedge$ on both $\mathcal{A}$ and $\mathcal{B}$, or a single $\psi$ that breaks $\vee$ on $\mathcal{A}$ and $\mathcal{B}$, we can make progress. In such cases, the hard instances constructed in Lemma~\ref{lem:marcin-after-bodirsky-et-al} are true on $\mathcal{A}$ iff they are true on $\mathcal{B}$, so the promise is fulfilled by definition. 
\begin{lemma}
Let $(\mathcal{A},\mathcal{B})$ be a pair of structures.
\begin{itemize}
\item If there exists $\phi$ that breaks $\wedge$ on both $\mathcal{A}$ and $\mathcal{B}$, then $\PMC(\mathcal{A},\mathcal{B})$ is \NP-hard. 
\item If there exists $\psi$ that breaks $\vee$ on both $\mathcal{A}$ and $\mathcal{B}$, then $\PMC(\mathcal{A},\mathcal{B})$ is \coNP-hard. 
\item If there exists $\phi$ and $\psi$ so that $\phi$ breaks $\wedge$ on both $\mathcal{A}$ and $\mathcal{B}$ and $\psi$ breaks $\vee$ on both $\mathcal{A}$ and $\mathcal{B}$,  then $\PMC(\mathcal{A},\mathcal{B})$ is \Pspace-hard.
\end{itemize}
\label{lem:marcin-after-bodirsky-et-al-promise}
\end{lemma}

Consider the template $(\mathcal{A},\mathcal{B})$ where the structures have three unary relations $U_1$, $U_2$ and $U_3$.
\[
\begin{array}{llll}
A=& \{1,2,3,4\} & B=& \{1,2,3,4,5,6\} \\
U^{\mathcal{A}}_1=& \{1\} & U^{\mathcal{B}}_1=& \{1,2,3\} \\
U^{\mathcal{A}}_2=& \{2,3\} & U^{\mathcal{B}}_1=& \{3,4,5\} \\
U^{\mathcal{A}}_3=& \{3,4\} & U^{\mathcal{B}}_1=& \{4,5,6\} \\
\end{array}
\]
Now, let us note that $\exists x \, (U_1(x) \wedge U_3(x))$ breaks $\wedge$ on both $\mathcal{A}$ and $\mathcal{B}$; while $\forall x \, ((U_1(x) \vee U_2(x)) \vee U_3(x))$ breaks $\vee$ on both $\mathcal{A}$ and $\mathcal{B}$.
\begin{corollary}
$\PMC(\mathcal{A},\mathcal{B})$ is \Pspace-complete.
\end{corollary}
This example would have been known to have been both \NP-hard and \coNP-hard from \cite{AsimiBB22}. However, its \Pspace-completeness is not covered in that paper, or by the subsequent work of these authors \cite{Kristina-personal} (including \cite{Asimi23}).

Let us consider another example, this time a listed open problem from \cite{AsimiBB22}. Let $(\mathcal{A},\mathcal{B})$ be the template where the structures have three unary relations $U_1$, $U_2$ and $U_3$.
\[
\begin{array}{llll}
A=& \{1,2,3\} & B=& \{1,2,3\} \\
U^{\mathcal{A}}_1=& \{1\} & U^{\mathcal{B}}_1=& \{2,3\} \\
U^{\mathcal{A}}_2=& \{2\} & U^{\mathcal{B}}_1=& \{1,3\} \\
U^{\mathcal{A}}_3=& \{3\} & U^{\mathcal{B}}_1=& \{1,2\} \\
\end{array}
\]
Now, let us note that $\exists x \, U_1(x) \wedge U_2(x) \wedge U_3(x)$ "breaks" $\wedge$ on both $\mathcal{A}$ and $\mathcal{B}$; while $\forall x \, (U_1(x) \vee U_2(x) \vee U_3(x))$ "breaks" $\vee$ on both $\mathcal{A}$ and $\mathcal{B}$. However, the "break" isn't on a single conjunction or disjunction (with two parts exactly), but rather on a triple. Furthermore, it can't be manipulated to be on a pair: e.g., $\exists x \, (U_1(x) \wedge U_2(x) \wedge U_3(x)$ breaks $\wedge$ on $\mathcal{B}$ but not on $\mathcal{A}$. We never defined breaking other than across conjunction or disjunction of \emph{pairs}. When it occurs across a triple such as this, we do not yet have the correct methods to prove \Pspace-hardness. Let us note that \NP-hardness and \coNP-hardness of $\PMC(\mathcal{A},\mathcal{B})$ follow from classical results from promise CSP \cite{AsimiBB22}. Among various remarkable properties of this template, let us note that $U^\mathcal{A}_i$  and $U^\mathcal{B}_i$ are set-theoretic complements, for each $i \in [3]$. 

\end{document}